\documentclass[12pt]{article}

\usepackage[usenames]{color}
\usepackage{amsfonts}
\usepackage{amssymb}
\usepackage{amscd}
\usepackage{amsmath}
\usepackage{amsthm}

\usepackage{graphicx}
\usepackage{wrapfig}
\newcommand{\BEQ}{\begin{equation}}
\newcommand{\EEQ}{\end{equation}}
\def\bea{\begin{eqnarray}}
\def\eea{\end{eqnarray}}
\def\nn{\nonumber}

\newtheorem{Th}{Theorem}
\newtheorem{Lem}{Lemma}
\newtheorem{Prop}{Proposition}

\newtheorem{Ex}{Example}
\newtheorem{Def}{Definition}

\def\bea{\begin{eqnarray}}
\def\eea{\end{eqnarray}}

\def\bes{\begin{equation*} \begin{split}}
\def\ees{\end{split} \end{equation*}}

\def\R{{\mathbb{ R}}}
\def\Z{{\mathbb{ Z}}}
\def\CC{{\mathbb{ C}}}

\usepackage{graphicx}
\begin{document}
\begin{flushright}
ITEP-TH-23/15
\end{flushright}
~\\
{\LARGE \bf Cohomology of the tetrahedral complex }
\vskip 1mm
~\\
{\LARGE \bf and quasi-invariants of $2$-knots}

\vskip 5mm
~\\
{\Large I.G. Korepanov\footnote{(MIREA) korepanov@mirea.ru}, 
G.I. Sharygin\footnote{(MSU, ITEP) gshar@yandex.ru}, D.V. Talalaev\footnote{(MSU, ITEP) dtalalaev@yandex.ru} }
\vskip 5mm
~\\
{\bf Abstract} 
\\
This paper explores a particular statistical model on 6-valent graphs with special properties which turns out to be invariant with respect to certain Roseman moves if the graph is the singular point graph of a diagram of a 2-knot. The approach uses the technic of the tetrahedral complex cohomology.
\\
We emphasize that this model considered on regular $3d$-lattices in the work \cite{T} appears to be integrable. We also set out some ideas about the possible connection of this construction with the area of topological quantum field theories in dimension $4$.

\tableofcontents
\vskip 2cm
\section{Introduction}
This work is one in the series of papers aimed to understand the transition from the Yang-Baxter equation to the Zamolodchikov tetrahedron equation. On the level of formulas their expressions are strikingly similar: the {\em matrix} Yang-Baxter equation takes the form
\bea
\label{YBE1}
R_{12}R_{13}R_{23}=R_{23}R_{13}R_{12} \in End(V^{\otimes 3}),\qquad R\in End(V^{\otimes 2}),
\eea
and the {\em matrix} tetrahedron Zamolodchikov equation \cite{Zam} - 
\bea
\label{TE1}
\Phi_{123}\Phi_{145}\Phi_{246}\Phi_{356}=\Phi_{356}\Phi_{246}\Phi_{145}\Phi_{123}\in End(V^{\otimes 6}),\qquad \Phi\in End(V^{\otimes 3}).
\eea
In both cases $V$  is a finite dimensional vector space, the indices denote the numbers of the space copies in which linear operators act non-trivially. One can also mention the $n$-simplex equation (eg. \cite{KST}), which is a straightforward generalization of both these equations to the case when the number of tensor factors on both sides of the equality grows. 

This affinity seems to be a specific manifestation of a ubiquitous more general phenomenon; similar patterns can be seen on the level of homotopy Lie algebras, combinatorial objects like braid group and the category of 2-tangles, integrable spin chains on 1- and 2-dimensional lattices, corresponding 2- and 3-dimensional statistical models and in many other places. In this paper we elaborate on possible relation between these equations and the low-dimensional topological objects such as 2-dimensional knots. Namely we make an attempt to apply the statistical model proposed in \cite{T} to construct certain invariants of $2$-dimensional knots. Unfortunately this attempt fails in full generality, the corresponding partition function appears to be invariant with respect only to some of Roseman moves which play in the theory of $2$-knots a r\^ole equivalent to the r\^ole of Reidemeister moves in $1$-knots. However the partial invariance of the constructed functional is not a totally trivial result and that is why we dare to explain it here. Moreover we expect some strong relations of our construction with $4$-dimensional topological quantum field theories like BF-theory; these hopes are inspired by the Jones-Witten invariant \cite{W} that relates the traditional 1-knot diagram technic with observables in the Chern-Simons theory. 

\subsection*{Acknowledgments}
The work of G.S was supported by RFBR grant 15-01-05990. 
The work of D.T. was partially supported by the RFBR grant 14-01-00012, grant of the support of scientific schools NSH-4833.2014.1, and the grant of the Dynasty Foundation.

\section{Preliminaries}
\subsection{Zamolodchikov tetrahedron equation}
\label{tetra}

We start with the recollection of the \textit{set-theoretic tetrahedral equation} (STTE) which may be described as follows: let $X$ be a finite set, we say that there is a solution for STTE on $X$ if there is a map
$$
X\times X\times X\stackrel{\Phi}{\longrightarrow} X\times X\times X,
$$
satisfying the relation
\begin{equation}
\label{eqtetr1}
\Phi_{123}\circ \Phi_{145}\circ \Phi_{246}\circ \Phi_{356} = \Phi_{356}\circ \Phi_{246}\circ \Phi_{145}\circ \Phi_{123}:X^{\times 6}\to X^{\times 6}.
\end{equation}
Here, however, unlike (\ref{TE1}), $X^{\times 6}$ denotes the Cartesian product of $X$ to itself six times, and the subscripts denote the number of factors to which $\Phi$ is applied, in other factors the map acts identically. For example
\begin{equation}
\label{eqtetr}
\begin{aligned}
\Phi_{356}(a_1,a_2,a_3,a_4,a_5,a_6)&=(a_1,a_2,\Phi_1(a_3,a_5,a_6),a_4,\Phi_2(a_3,a_5,a_6),\Phi_3(a_3,a_5,a_6))\\&=(a_1,a_2,a_3',a_4,a_5',a_6'),
\end{aligned}
\end{equation}
where
$$
\Phi(x,y,z)=(\Phi_1(x,y,z),\Phi_2(x,y,z),\Phi_3(x,y,z))=(x',y',z').
$$
Graphically equation \eqref{eqtetr1} can be explained by figure \ref{TE1}, which describes the evolution of the entries in Cartesian product on both sides of the equality.
\begin{figure}[h]
\center{\includegraphics[width=15cm]{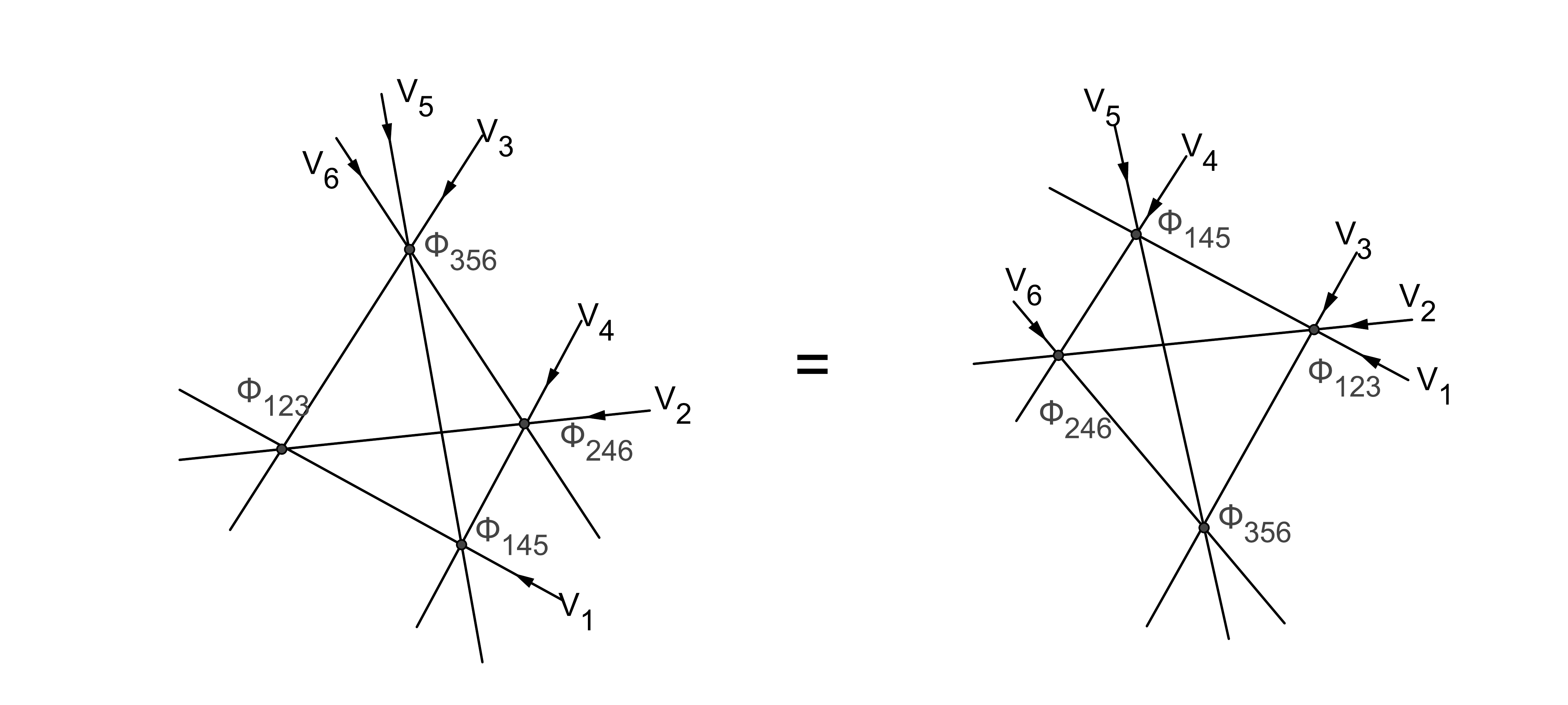}}
\caption{Tetrahedron equation}
\label{TE1}
\end{figure}
There is an important particular class of solutions of this equation, which can be obtained by a suitable dualization of the functional tetrahedral equation, i.e. of the usual (matrix) Zamolodchikov tetrahedron equation in which the vector space $V$ is the space of functions on a (algebraic) variety; for example one can take the so-called electric solution of functional equation, which is given by the following transformation (or ansatz) acting on the space of functions of three variables:
\bea
\label{el}
\Phi(x,y,z)=(x_1,y_1,z_1);\\
x_1=\frac {x y} {x+z+x y z},\nn\\
y_1=x+z+x y z,\nn\\
z_1=\frac{ yz} {x+z+x y z}.\nn
\eea
This solution is relevant to the well-known star-triangle relation in the theory of electric circuits.

Another important interpretation of the tetrahedron equation is in terms of the combinatorial problem of coloring of $2$-faces of the $4$-dimensional cube with elements of some set $X$. Let $\Phi:X\times X\times X \rightarrow X \times X \times X$ be a map. A coloring of 2-faces of $I^4$ is a map from the set of 2-faces into $X$; it is called admissible if the colors $x,y,z$ of the incoming faces of any $3$-face are related with the colors $x',y',z'$ of its outgoing faces by the action of the map $\Phi:$
\bea
(x',y',z')=\Phi(x,y,z)\nn
\eea
The problem is discussed in more detail in \cite{KST}.

\begin{figure}[h]
\begin{minipage}[h]{0.49\linewidth}
\center{\includegraphics[width=1\linewidth]{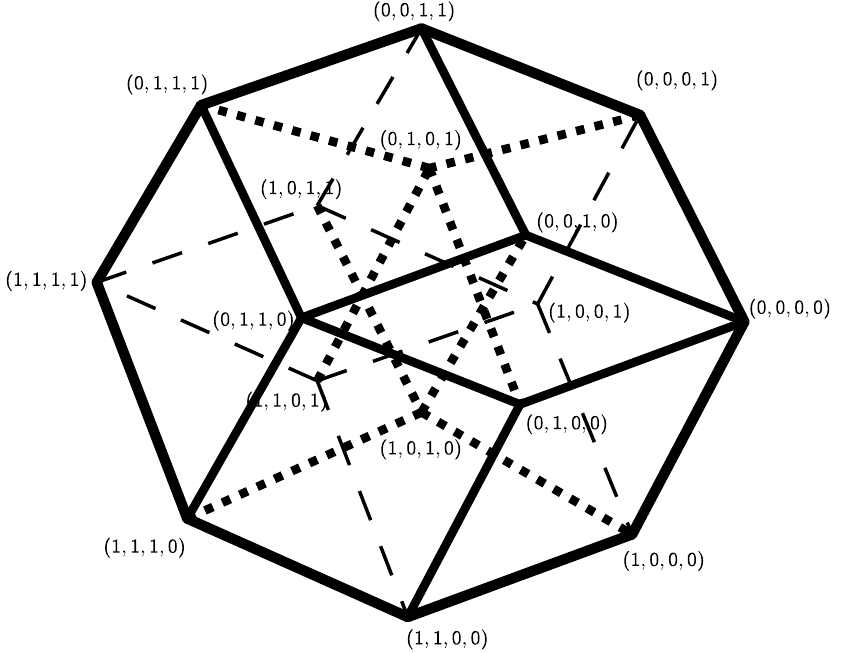} }
\caption{Tesseract}
\label{tesseract}
\end{minipage}
\hfill
\begin{minipage}[h]{0.49\linewidth}
\center{\includegraphics[width=0.7\linewidth]{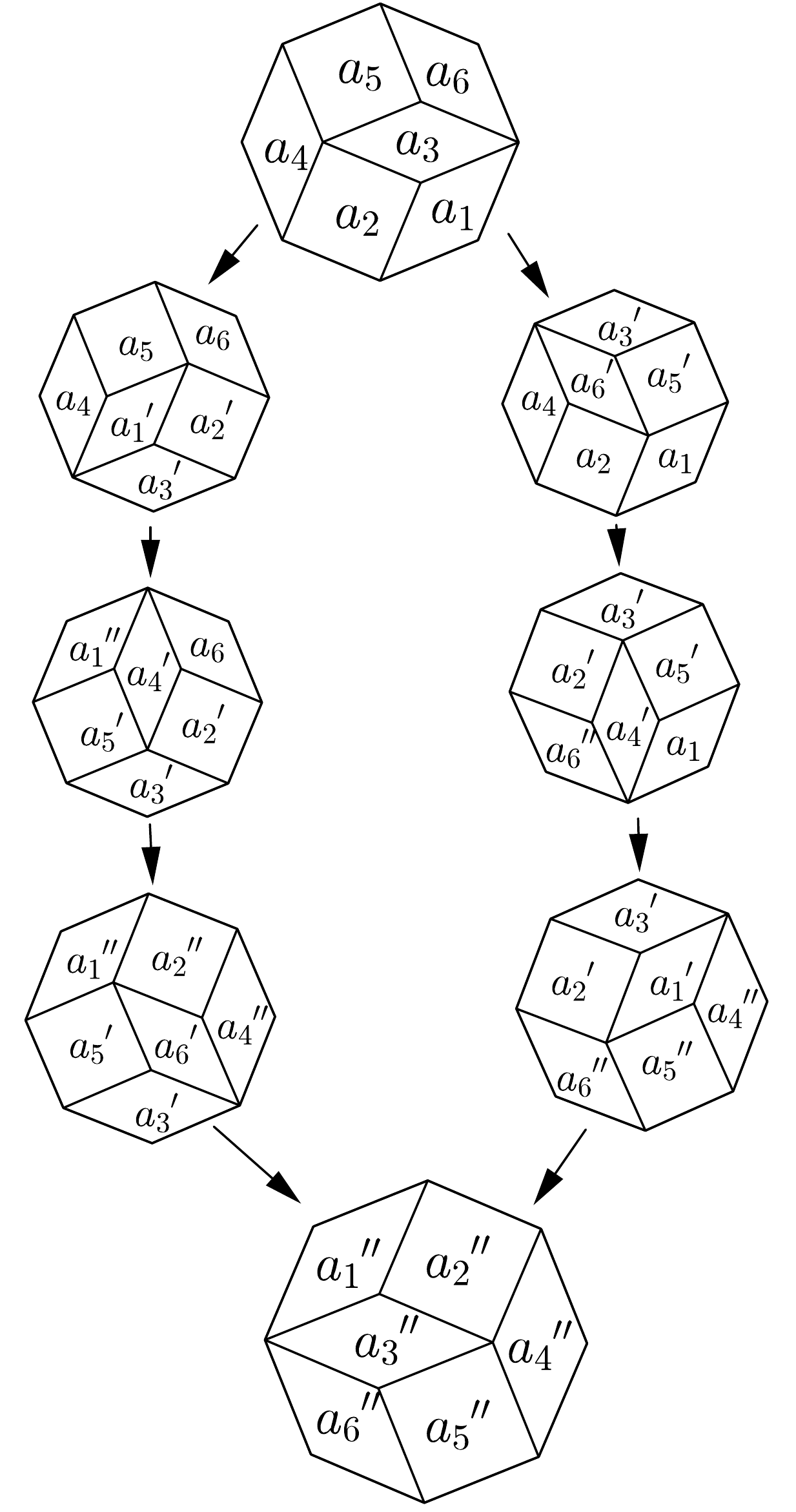} }
\caption{Coloring}
\label{proc}
\end{minipage}
\end{figure}

Figure \ref{tesseract} shows a projection of the $4$-cube to a $3$-space and  figure \ref{proc} shows two possible ways in which the faces can be colored, starting from visible faces and ``pushing across'' one 3-face after another. It turns out that the colorings obtained by these two ways are equal if and only if the tetrahedral equation on $\Phi$ holds.

Similar colorings of $2$-faces of an arbitrary $N$-cube give rise to a homology and cohomology complexes, associated with a solution of the set-theoretic tetrahedral equation. These complexes are completely analogous to those calculating the Yang-Baxter cohomology, see \cite{CES}. We give some detail on this subject in the following section.

\subsection{Tetrahedral complex}
\label{sectimp}
Let us recall some results from \cite{KST}. Let $I^N$ be the standard $N$-cube, i.e. 
\[
I^N=\underbrace{I\times I\times\dots\times I}_{N\ \mbox{times}},\ \mbox{where}\ I=[0,1].
\]
We shall denote by $I^N_2$ the set of its $2$-faces. One can parametrize $I^N_2$ by sequences of symbols $\tau=(\tau_1,\ldots,\tau_N)$ where $\tau_k$ take values $0,1,\ast$; here $\ast$ corresponds to a coordinate varying in the interval  $[0,1]$. It follows from the definitions, there are exactly $n$ asterisks in a sequence, corresponding to an $n$-dimensional face.

Let us fix a face $\tau$ and denote by $\{j_k\}_\tau$ the set of indices of symbols $\ast$ in the sequence $\tau$. A codimension $1$ subface of $\tau$ is determined by substituting one of the numbers $0$ or $1$ instead one of the asterisks in the sequence. Let us fix the index $j_k$ of the corresponding symbol.

We define the alternating sequence:
\bea
\varkappa_{1}=0,\varkappa_{2}=1,\varkappa_{3} \ldots.\nn
\eea
\begin{Def}
 The subface is called incoming if the $j_k$-th coordinate of the subface coincides with $\varkappa_k$ and outgoing otherwise. 
\end{Def}
\begin{figure}[h!]
\centering
\includegraphics[width=6cm]{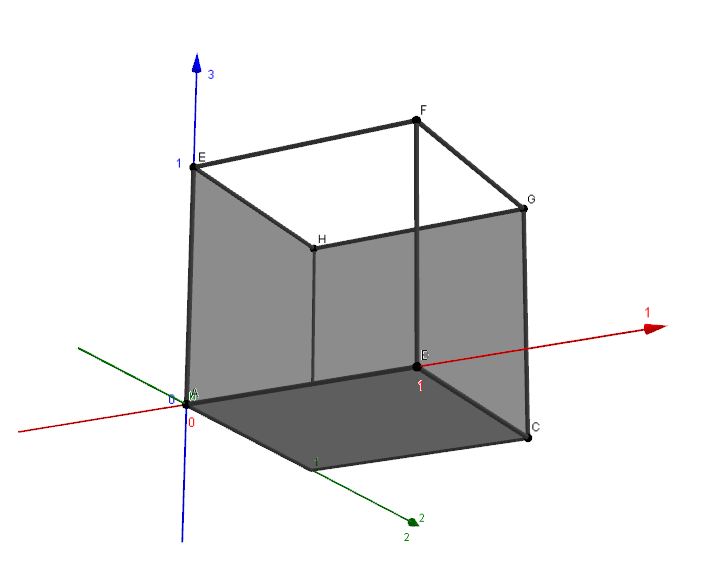}
\caption{Incoming(black) and outgoing(white) faces of a standard $3$-cube}
\end{figure}
Let us fix a set $X$ and a solution of the set-theoretic tetrahedron equation $\Phi:X\times X\times X\rightarrow X\times X\times X$.
\begin{Def}
A coloring of  $2$-faces of an $N$-cube $c:I^N_2\rightarrow X$ is called admissible if for any  $3$-face the colors of its incoming $2$-faces $(x,y,z)$ and the colors of the corresponding outgoing $2$-faces $(x',y',z')$ are related by the equality
\bea
(x',y',z')=\Phi(x,y,z),\nn
\eea
where the order of faces is defined in the lexicographic way.
\end{Def}
Let us consider a complex $C_*(X)=\bigoplus_{n\ge2}C_n(X)$ where
$$
C_n(X)=C_n(X,\,\Bbbk)=\Bbbk\cdot C^2(n,\,X),
$$
here $C^2(n,\,X)$ is the set of all admissible colorings $2$-faces of an $n$-cube, and $\Bbbk\cdot C^2(n,\,X)$ is the free $\Bbbk$\/-module generated by it. The differential $d_n:C_n\to C_{n-1}(X)$ is defined by the formula  
$$
d_n(c)=\sum_{k=1}^n\left(d^i_kc-d^o_kc\right),
$$
where $d^f_ic$ ( $d^o_kc$) is the restriction of the coloring $c$ to the  $k$\/-th incoming (resp. outgoing) $(n-1)$\/-face of the cube $I^n$. Denote by $H_*(X,\,\Bbbk)$ the corresponding homologies; we call it \textit{tetrahedral homology of $X$ with coefficients $\Bbbk$}.
\begin{Def}
We call an $n$-face of an $N$-cube absolutely incoming if it is not outgoing of any $n+1$-face.
\end{Def}
\begin{Lem}
 A coloring of $2$-faces of an $N$-cube is uniquely defined by a coloring of absolutely incoming $2$-faces.
 \end{Lem}
 The number of absolutely incoming $2$-faces is equal to $C_N^2$. Hence in low dimension the complex is represented by
$$
\begin{aligned}
C_2(X)&=\Bbbk\cdot X,\\
C_3(X)&=\Bbbk\cdot X^{\times3},\\
C_4(X)&=\Bbbk\cdot X^{\times6}.
\end{aligned}
$$
We will usually denote a coloring by the corresponding colors of absolutely incoming faces; this notation is short, but formulas for the differential in this notation become rather complicated.

Further, we can dualize the construction and consider the \textit{tetrahedral cohomology of $X$}. By definition an element of the corresponding complex is equal to a linear map from $C_*(X)$ into $\Bbbk$; this map is uniquely defined by its values on the basis elements, i.e. on colorings, and the differential $\partial$ is equal to the conjugation of $d:C_n(X)\to C_{n-1}(X)$. 

\textit{Below we shall usually consider the \textbf{potentiated version of tetrahedral cohomology}}, i.e. let $\Bbbk=\mathbb R$ or $\mathbb C$. For any cochain $f:C_n(X)\to\Bbbk$, we can associate the map $\phi=\exp{(f)}:C_n(X)\to\Bbbk$, so that in all the formulas, including the formula for the differential all the sums will be replaced by products. We shall call the corresponding formulas \textit{multiplicative}. Clearly, such formulas can be written for an arbitrary field $\Bbbk$ (although in a general case there will not be such a clear relation with the usual cohomology). An important example is when $n=4$; in this case the equality $\partial f=0$ implies the following equation for the multiplicative $3$-cocycle:
 \bea
\label{cocycle}
&\varphi(a_1,a_2,a_3)\varphi(a_1',a_4,a_5)\varphi(a_2',a_4',a_6)\varphi(a_3',a_5',a_6')=\nn\\
 &=\varphi(a_3,a_5,a_6)\varphi(a_2,a_4,a_6')\varphi(a_1,a_4',a_5')\varphi(a_1',a_2',a_3')
\eea
in the notation of the picture \ref{proc}. Then the following simple but important statement holds (Theorem 4 \cite{KST}):
\begin{Lem}
\label{lem_lin}
Let $\Phi$ be a solution for the STTE and let $\phi$ be a multiplicative $3$-cocycle of the tetrahedral complex. Denote by $V=V(X)$ the vector space generated by the elements of the set $X$. Then let us define a linear operator $A$ on $V^{\otimes 3}$ specifying its values on tensor products of basis vectors. We say that 
\[
A(s)(e_x\otimes e_y \otimes e_z)=\phi(x,y,z)^s(e_{x'}\otimes e_{y'}\otimes e_{z'})
\] 
if and only if $\Phi(x,y,z)=(x',y',z')$ (here $e_x$ denotes the basis element in $V$, that corresponds to some $x\in X$). In this case $A(s)$ provides a solution for the matrix tetrahedral equation.
\end{Lem}

The importance of this lemma is not only in that it allows one produce many examples of solutions of the usual (matrix) tetrahedron equation, but also in the following observation: given a map of sets $F:X\times X\to X\times X$, one can ask about the map $\tilde F:X\to X$, obtained by ``substituting the values of the first leg of $F$ as the second argument of $F$''. In other words, let for every $x\in X$ 
\[
F_x:X\to X\times X,\ F_x(y)=F(x,y).
\]
Projecting at the cartesian factors we can write this map as $F_x=(F_{x,1},F_{x,2})$. Then we look for the solution $y_x$ of the equation 
\[
F_{x,1}(y)=y,
\]
and put $\tilde F(x)=F_{x,2}(y_x)$. Below we shall need a more complicated version of this construction for the solution of tetrahedron equation map $\Phi:X\times X\times X\to X\times X\times X$, where we shall obtain in this way a map $\tilde\Phi:X\times X\to X\times X$. Observe, that we can, in fact, perform this operation in $6$ different ways. All these procedures are rather complicated, but for the linear map $A(s)$, the corresponding analogue is given simply by the convolution of a pair of its indices, one upper, and one lower. Clearly there are 6 different ways to do this convolution: for every pair of indices $i,j=1,2,3,\ i\ne j$ we can choose either $i$, or $j$ to be the lower index, and the other way round, just like there are 6 ways to insert one of the values of the map $\Phi$ as an argument into another leg. Let us denote such convolution by $A(s)^i_j$.

Keeping this in mind, \textit{we shall say, that the cocycle $\phi$ is \textbf{completely normalized}, if the convolution of $A(s)$ in any two indices is equal to the identity map}. More accurately, we can say, that \textit{$\phi$ is normalized in direction $(i,j),\ i,j=1,2,3,\ i\ne j$, if $A(s)^i_j=\mathrm{Id}$ and $A(s)_i^j=\mathrm{Id}$}. An example of such partially normalized map is given in section \ref{seq_coc}.

\subsection{2-knots}
\begin{Def}
Here and below \textit{we call a $2$-knot an isotopy class of embeddings $S^2\hookrightarrow {\mathbb R}^4.$}
\end{Def}
A class of examples of non-trivial $2$-knots is given by the Zeeman's twisted-spun knot construction \cite{Zeeman}, which is a generalization of the Artin spun knot and is illustrated at the figure \ref{zee}.
\begin{figure}[h!]
\centering
\includegraphics[width=50mm]{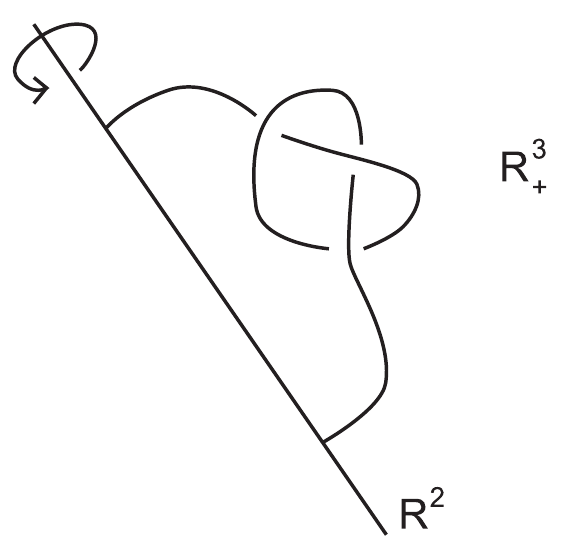}
\caption{Zeeman knot}
\label{zee}
\end{figure}
Here a $k$-times rotation of the long knot in $R_+^3$ around its axis is supposed.

To obtain a diagram of a $2$-knot one takes a generic projection $p$ from $\mathbb R^4$ to a hyperplane $\mathbb R^3\cong P\subset{\mathbb R}^4.$ The generic position means that there are only simple singularities of the composite map, i.e. only of the singularities of the following three types: double point, triple point and the Whitney point (or branch point) see figure \ref{singtypes}.
\begin{figure}[h!]
\centering
\includegraphics[width=3cm]{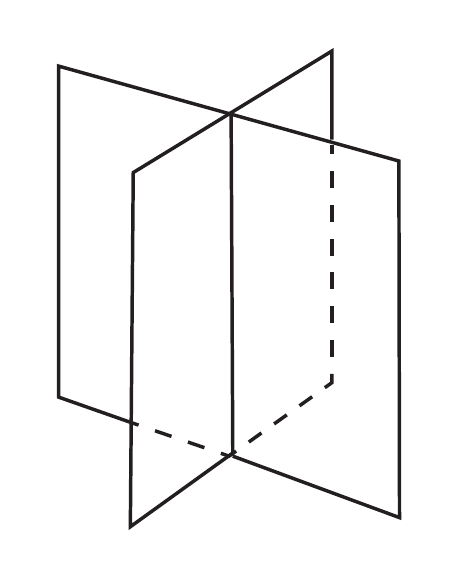}
\includegraphics[width=4cm]{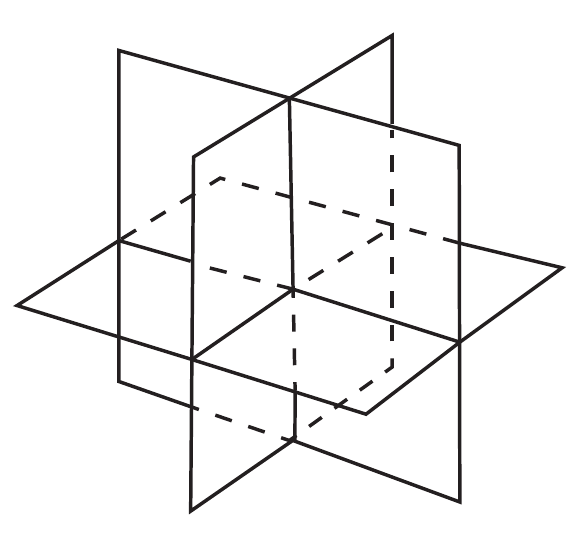}
\includegraphics[width=4cm]{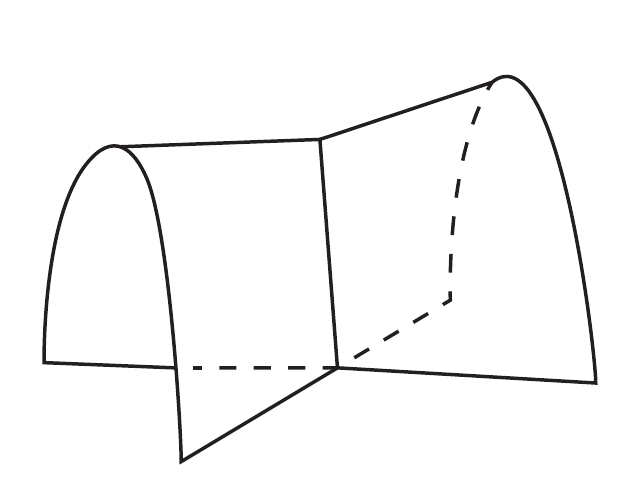}
\caption{Singularity types}
\label{singtypes}
\end{figure}
The diagram of a $2$-knot is a singular surface with arcs of double points which end in triple points and branch points. This defines a graph of singular points. The additional information consists of the order of $2$-leaves in intersection lines subject to the projection direction (i.e. according to the distance from these leaves to the hyperplane $P$). We shall also always assume that the surface we consider is oriented (e.g. the positive normal at every point is given). In particular this allows one associate a sign with every triple point, by saying that the point is positive, if the frame given by positive normals to the 2-dimensional leaves of the diagram that meet in the point, is positive oriented, and the point is negative otherwise. 

We say, that \textit{two diagrams are isotopic}, if they can be connected by a continuous (even smooth) family of diagrams with simple singularities. Of course, in general the same knotted surface in $\mathbb R^4$ can be represented by two non-isotopic diagrams, i.e. two diagrams, which can be connected only by families that contain higher degree singularities. Thus in order to tell, which diagrams represent the same knot, we need to describe admissible higher singularities, that can appear in a generic smooth family of diagrams. The following theorem gives a complete list of such possible transformations; it is an analog of the well-known Reidemeister moves in knot theory:
\begin{Th}[Roseman \cite{Roseman98}]
Two diagrams represent equivalent knotted surfaces iff one of them can be obtained from another by a finite series of moves from the list given in figures \ref{13RM}, \ref{46RM}, \ref{7RM} and an isotopy of a diagram in $\mathbb{R}^3.$ 
\end{Th}
\begin{figure}[h!]
\centering
\includegraphics[width=35mm]{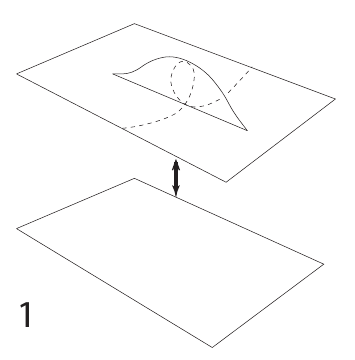}
\includegraphics[width=35mm]{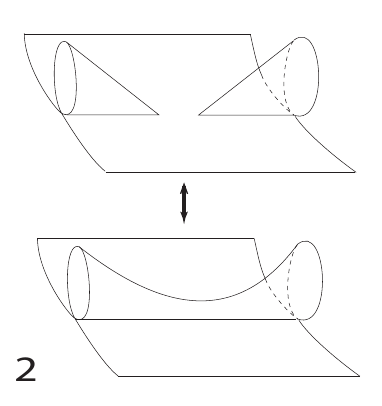}
\includegraphics[width=70mm]{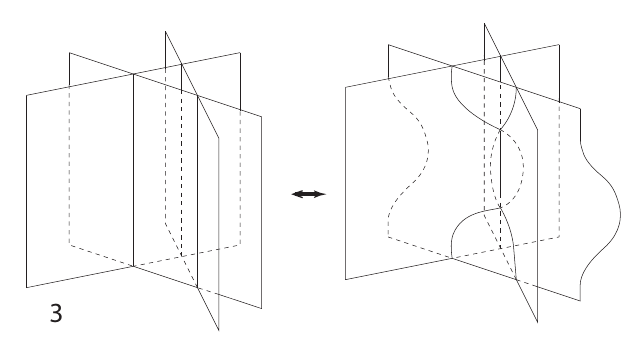}
\caption{1-3 Roseman moves}
\label{13RM}
\end{figure}
\begin{figure}[h!]
\centering
\includegraphics[width=40mm]{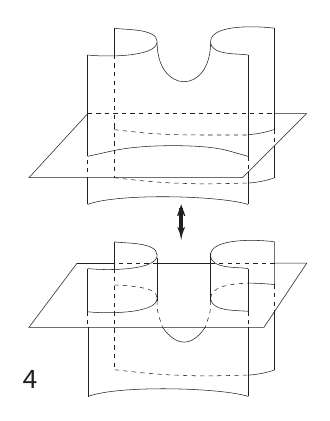}
\includegraphics[width=40mm]{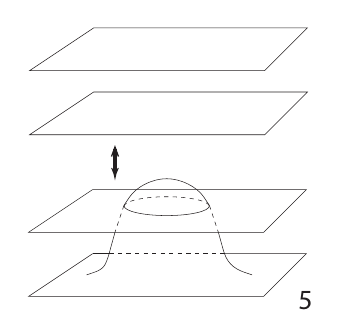}
\includegraphics[width=35mm]{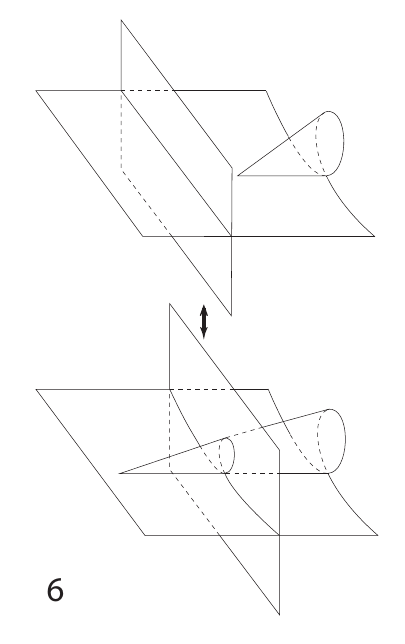}
\caption{4-6 Roseman moves}
\label{46RM}
\end{figure}
\begin{figure}[h!]
\centering
\includegraphics[width=80mm]{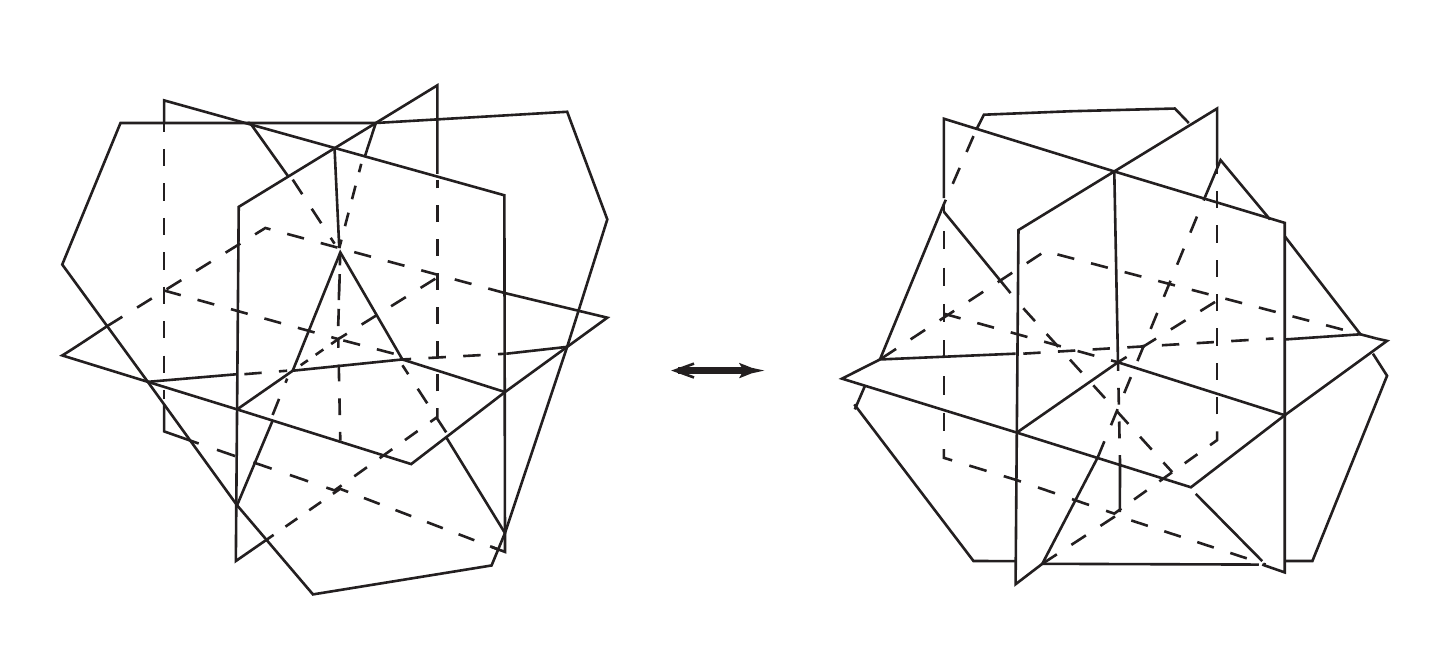}
\caption{7th Roseman move}
\label{7RM}
\end{figure}
\subsection{Quandle cohomology and 2-knot invariants}
There is an approach due to Carter, Saito and others \cite{CSJ} which produces invariants of 2-knots by means of the so called quandle cohomology (more generally, one can use Yang-Baxter maps and cohomology for the same purposes, see \cite{CES}, the there construction is quite similar to the one, based on quandles, so we restrict our attention here only to the latter).

In this theory invariants of 2-knots are constructed as some partition functions on the space of states which are colorings of the 2-leaves of a diagram by elements of a quandle. So let us first recall the definition of \textit{quandle}:
\begin{Def}[Matveev \cite{Mat82}]
A set $X$ with a binary operation $(a,b)\mapsto a\ast b$ is called quandle if
\bea
i)&& \forall a\in X\quad a\ast a=a\nn\\
ii)&& \forall a,b \in X \quad \exists ! c\in X: c\ast b=a\nn\\
iii)&& \forall a,b,c \in X \quad (a\ast b)\ast c=(a\ast c)\ast(b\ast c)\nn
\eea
\end{Def}
\begin{Ex}
The simplest example of such structure is given by a group quandle, i.e. when $X$ is the set of group elements $G$ with the operation $a\ast b=b^{-n} a b^n$  for any fixed $n.$
\end{Ex}
\begin{Ex}
Another important example is the Alexander quandle: $X$ is a $\Lambda$-module $M$, where $\Lambda=\Z[t,t^{-1}]$, with the operation $a\ast b=t a+(1-t)b.$
\end{Ex}
One can define the quandle cohomology as follows. Let us define a complex  $C_n^R(X)$ whose components are free abelian groups generated by $n$-tuples of elements of  $X$ $(x_1,\ldots,x_n).$ Then the differential  $\partial_n:C_n^R(X)\rightarrow C_{n-1}^R(X)$ is:
\bea
\partial_n(x_1,\ldots,x_n)&=&\sum_{i=2}^n(-1)^i\{(x_1,x_2,\ldots,x_{i-1},x_{i+1},\ldots,x_n)\nn\\
&-&(x_1\ast x_i,x_2\ast x_i,\ldots,x_{i-1}\ast x_i,x_{i+1},\ldots,x_n)\}\nn
\eea
Consider a subcomplex $C_n^D(X)$, whose components are generated by $n$-tuples $(x_1,\ldots,x_n)$ with $x_i=x_{i+1}$ for some $i$ and $n\ge 2$. We construct a quotient complex $C_n^Q(X)=C_n^R(X)/C_n^D(X)$ with induced differential. Then the homology and cohomology of quandle $X$ with coefficients in an Abelian group $G$ are by definition, the (co)homology of the complexes:
\bea
C_{\ast}^Q(X,G)=C_{\ast}^Q(X)\otimes G,&&\partial=\partial\otimes id\nn\\
C_{Q}^{\ast}(X,G)=Hom(C_{\ast}^Q(X), G),&&\delta=Hom(\partial, id)\nn
\eea
Now one can define the invariant of a 2-knot as follows: first we shall cut the diagram along the double points set into a collection of non-singular surfaces with boundary in $\mathbb R^3$, saying that whenever two leaves meet in a double point, the lower one is cut by the upper one (lower and higher denote their distances from the plane $P$, see above). These surfaces will be called \textit{2-leaves of the diagram}. The we define the diagram coloring: we denote by $L$ the set of 2-leaves of a diagram after cutting. One says that there is a coloring $C$ of a diagram $D$ with elements of a quandle $Q$ if there is a map $C:L\rightarrow Q$ satisfying the coherence conditions near the intersections of the diagram illustrated by the picture:
\begin{figure}[h!]
\centering 
\includegraphics[width=45mm]{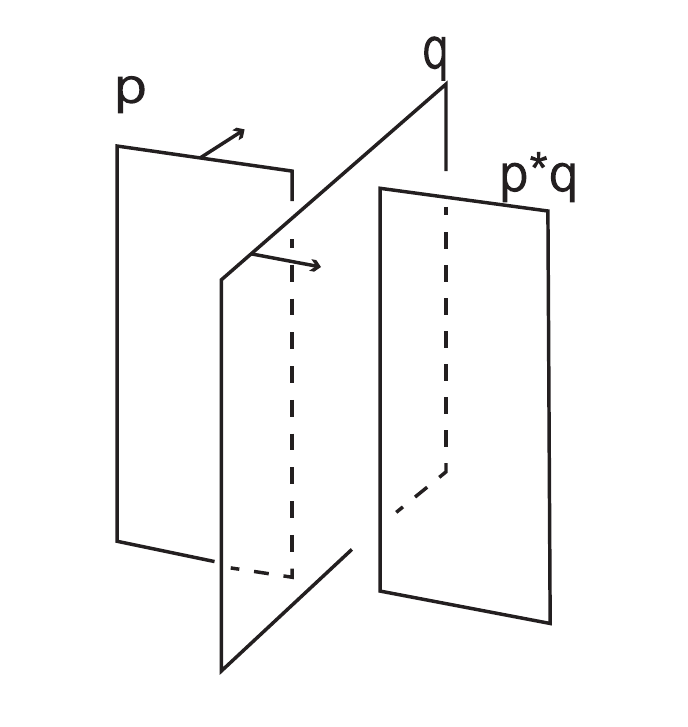}
\caption{Coloring}
\label{coloring}
\end{figure}
Let us fix a $3$-cocycle $\theta\in Z_Q^3(Q,A).$ This implies a condition
\bea
\theta(p,r,s)+\theta(p\ast r,q\ast r,s)+\theta(p,q,r)=\theta(p\ast q,r,s)+\theta(p,q,s)+\theta(p\ast s,q\ast s,r\ast s)\nn
\eea
One attributes the following Boltzmann weight to a triple point $\tau$ 
\bea
B(\tau,C)=\theta(x,y,z)^{\epsilon(\tau)}\nn
\eea
where $\epsilon(\tau)$ is the sign of  the triple point $\tau$ (determined by the normals), $x,y,z$ - colors of the top, middle and bottom leaves in outgoing octant, i.e. such that it is opposite to normals of all leaves. Then one defines the partition function
\bea
\label{quandle-stat}
S(D, \theta,A)=\sum_C\prod_\tau B(\tau,C).
\eea
\begin{Th}[Carter et al., \cite{CSJ}]
The partition function \ref{quandle-stat} is invariant with respect to the Roseman moves and hence is an invariant of a $2$-knot.
\end{Th}

\section{Quasi-invariants of $2$-knots}
\subsection{Statistical model}
\label{seq_coc}
In this section we are going to construct the principal object of the present paper: the function $\chi(D)$ (where  $D\looparrowright\mathbb R^3$ is the diagram of a 2-knot $\Sigma\hookrightarrow\mathbb R^4$), which is invariant under certain Roseman moves. We shall assume, that the knot is oriented (in the sense, that the embedded surface is oriented). In fact, one can think that this is just a sphere $S^2$. We shall also fix an orientation of the Euclidean space $\mathbb R^3$.

We begin with the following observations (similar to the discussion, given in the section about quandle invariants, see above): for every diagram of a 2-knot, the double-point set is an embedded graph in $\mathbb R^3$ of a very special kind; all its vertices have degree 6 or 1, where vertices of degree 6 correspond to the triple points of the diagram, and vertices of degree 1 are the singular points (``summits'' of Whitney's umbrellas). Moreover, for every vertex of degree 6, one can divide the adjacent edges into pairs (we shall call a pair of edges \textit{a line}), each of which can be regarded as a single branch of the double points-set.

Further, we choose a unit normal vector $\vec n_x$ at every double point $x\in D$ so that the orientation of the frame $\vec e_1,\,\vec e_2, \vec n_x$ (where $\vec e_1,\,\vec e_2$ is an oriented base of $D$ at $x$) in $\mathbb R^3$ is positive. Now since the knot and $\mathbb R^3$ is oriented, one can choose direction for every edge $E$ of the double points graph: every edge $E$ is equal to the intersection of two ``branches'' $F_1$ and $F_2$ of the surface, so that at every point $x\in E$ we have two normal vectors $\vec n_1\,\vec n_2$, determined by these branches. Recall now, that the diagram $D$ is equal to the image of $\Sigma$ under projection $p:\mathbb R^4\to\mathbb R^3$, where $\mathbb R^3$ is a suitable hyperplane. Thus one can speak about the distance from a point $x\in\Sigma$ to the plane. So, for every double point $x$ we know, which branch is nearer to the image, and which one is farther. Let $F_1$ be the nearer one, and $F_2$ the farther one; we assume that $\vec n_1$ will always precede $\vec n_2$. Now we can choose the direction $\vec v$ of the edge $E$ (one should think of $\vec v$ as of a unit vector along the edge) so that the frame, formed by $\vec n_1,\,\vec n_2$ and $\vec v$ is positive oriented. Observe, that this orientation does not change, when one passes through a triple point.

Finally we introduce the orientation of triple points of the diagram. As we have discussed it before, one can associate with it a sign. Namely, let $F_1,\,F_2$ and $F_3$ be the branches of the diagram, meeting in the triple point $O$; they are ordered according to the distance from the image. We shall order the lines, that meet in $O$ so that the line, transversal to the $i$-th branch has number $i$. We further introduce the sign of the point $O$: it is $+1$, if the normal vectors $\vec n_1,\,\vec n_2,\,\vec n_3$ to the intersecting branches, taken in the order, prescribed by the distances from the plane $\mathbb R^3$, give a positive-oriented frame, and it is $-1$ otherwise; we shall denote the orientation by $\sigma(O)$. Observe, that we obtain the same orientation, if we use the directions along the lines, meeting in $O$, ordered as explained before.

Now we can describe the construction of the semi-invariant function $\chi$. As a matter of fact, it can be defined in terms of the double-points graph of the diagram. So, suppose we are given a graph $\Gamma\subset\mathbb R^3$ with vertices of degrees $1$ and $6$, such that the edges at every point of degree 6 are divided into three lines, and each line is oriented. We also choose the order of lines at a triple point and orient the point as before.

Let $(X,\,\Phi)$ be a solution of the set-theoretical tetrahedron equation on a finite set $X$ and $h=\#X$ be its cardinality (infinite sets ask for additional assumptions); we shall assume that $\Phi$ is invertible. We begin with definition of a \textit{coloring} of the graph $\Gamma$ into $X$ colors. 
\begin{Def}
We shall say, that a \textbf{regular coloring} (or just a \textbf{coloring}) $c(\Gamma)$ of $\Gamma$ is given, if there is an element $x_e\in X$ assigned to every edge $e$ of $\Gamma$ so that for every degree $6$ point $O\in\Gamma_6$ the elements assigned to its outgoing edges are determined by the elements on incoming edges by the rule:
\[
(x',y',z')=\Phi^{\sigma(O)}(x,y,z).
\]
\end{Def}
Observe, that in the case, when $\Phi^{-1}=\Phi$ orientation of the vertex does not interfere into the definition.

Let now $\phi$ be a multiplicative $3$-cocycle in the tetrahedral cohomology with coefficients in the field of real (or complex) numbers, we have introduced earlier. Then we associate to every triple point $O$ of the diagram the number $\phi^{\sigma(O)}(x_O,y_O,z_O)$, where $(x_O,y_O,z_O)$ are the colors of incoming edges (the order given by the order of the edges), if $\sigma(O)=1$ and they are the colors of outgoing edges otherwise.

Finally let $s$ be a (real or complex) number, we put
\begin{equation}
\label{eq:definv}
\chi_\phi(s;\Gamma)=h^{-d}\sum_{c(\Gamma)}\prod_{O\in\Gamma_6}\phi^{\sigma(O)s}(x_O,y_O,z_O),
\end{equation}
Here $d$ is the number of connected components of the graph of singular points of $\Gamma(D)$ (and $h$ is the cardinality of $X$, see above).
\begin{Def}
The quasi-invariant $\chi_\phi(s;\Sigma)$ of a two-dimensional knot $\Sigma$ is given by $\chi_\phi(s;\Gamma(D))$, where $D$ is a diagram of the knot and $\Gamma(D)$ denotes the double points graph of $\Sigma(D)$.
\end{Def}

In what follows, we shall need if defined \textit{partly-transposed operators}, associated with the tetrahedron map $\Phi$: for instance, we shall say, that 
\[
\Phi^{t_2}(x,y,z)=(x',y',z'),
\]
if 
\[
\Phi(x,y',z)=(x',y,z').
\]
One can consider similar partly-transposed operators for arbitrary direction or combination of directions; we shall denote them by $\Phi^{t_i}$ or $\Phi^{t_it_j},\ i,j=1,2,3$. Clearly such operators do not exist for all $\Phi$; thus we introduce an additional condition on $\Phi$: \textit{we shall assume that all partly-transposed operators exist}. Observe, that $\Phi^{t_1t_2t_3}=\Phi^{-1}$; also observe that the identity $(\Phi^{t_i})^{t_i}=\Phi$ implies that $(\Phi^{t_i})^{-1}=\Phi^{t_jt_k},\ \{j,k\}=\{1,2,3\}\setminus \{i\}$.

\begin{Prop}
Let $X$ be a finite set and $\Phi:X^{\times 3}\to X^{\times 3}$ be a solution of the tetrahedron equation; let also $\phi\in C_3(X)$ be a multiplicative 3-cocycle. Then the expression $\chi_\phi(s;\Sigma)$ does not change when the diagram is modified by the 1-st, 3-rd, 5-th, and 7-th Roseman moves. If, in addition, the cocycle is normalized in the directions $(1,2)$ and $(2,3)$ (in particular, if it is completely normalized), then this value is not changed by the 6-th Roseman move.
\end{Prop}
First of all observe that the invariance with respect to the 1-st and 5-th moves is quite obvious; it follows from the fact that these moves change the number of connected components, thus the sum in the definition of value of $\chi_\phi(s;\Sigma)$ is just multiplied by the cardinality of the set $X$, so that the final result is preserved, since the first factor is divided by $h$. 

\subsection{The 6-th move}
In order to prove the invariance of the function $\chi_\phi(s;\Sigma)$ with respect to the 6-th move for normalized cocycles, it is enough to observe, that in the transformed diagram two of the three lines at the triple point are connected by a loop. Thus, in the sum $\chi_\phi(s;\Sigma)$ we shall have to add the terms, corresponding to the coloring of the new triple point, colored by the map $\tilde\Phi$, obtained from $\Phi$ by ``substituting the value of one of the legs of $\Phi$ as the argument into another'' (see remarks in the end of the section \ref{sectimp}), see figure \ref{6-move}. When we pass to the level of $\chi_\phi(s;\Sigma)$, these terms correspond to the multiplication of previous expressions by matrix elements of the map $A(s)^i_j$, where the convolution is taken over the indices, corresponding to the loop. Thus, the sum does not change, when the cocycle is completely normalized. More accurately, since the vertical 2-dimensional sheet cannot go between the other two branches, that meet in this triple point (these branches are, in fact, parts of the same sheet forming the Whitney umbrella), it is enough to ask only that $\phi$ should be normalized in the directions $(1,2)$ and $(2,3)$. In both case, the factor associated with the triple point is equal to 1, so the partition function is equal to the one, when there is no triple point, see figure \ref{6-move22}.

\begin{figure}[h!]
\centering
\includegraphics[width=6cm]{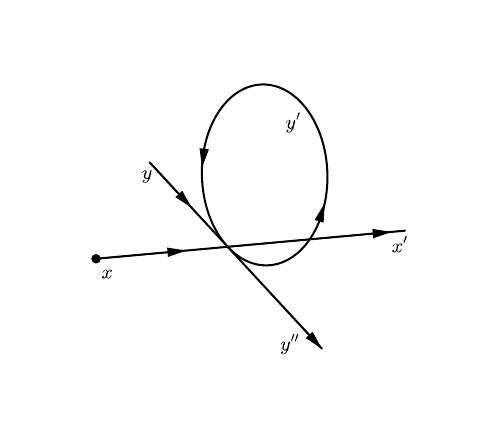}
\caption{The loop in the tetrahedral map}
\label{6-move}
\end{figure}

Let us give an example of normalized cocycles. First of all we recall that according to lemma 6 of \cite{KST} the electric solution can be used to produce the solutions of STTE on finite sets. To this end consider the residue ring~$\Z/p^k\Z$, where $p$ is either a prime number of the form $p=4l+1$, or $p=2$, and $k$ is an integer $\ge 2$. For such~$p$ the Legendre symbol $\left(\frac{-1}{p}\right)$ is equal to $1$, that is, there exists a square root of $-1$ in~$\Z/p\Z$. We fix one such root and call it $\varepsilon\in \Z/p\Z$. Our~$X$ will be the following subset of~$\Z/p^k\Z$:
\begin{equation}\label{X}
X=\{x \in \Z/p^k\Z\,\colon\;\; x=\varepsilon \mod p\}.
\end{equation}
Then the electric solution restricts well to this subset. We also stress the fact that $X$ is closed under inversion of elements in $\Z/p^k$. Let us analyze the set of colorings of a diagram at picture \ref{6-move} with respect to such a reduced electric solution. We start with the order $(x,y,y').$ The condition that should be verified at this triple point is:
\bea
 x'&=&xy/(x+y'+xyy'),\nn\\
 y'&=&x+y'+xyy',\nn\\
 y''&=&yy'/(x+y'+xyy').\nn
 \eea
 Solving the second equation one obtains $yy'=-1.$ Then substituting this to the other equations entails: $x'=-xy^2, y''=y.$ Another nonequivalent configuration corresponds to the order $(x,y',y).$ The following is the condition for the coloring in this case:
 \bea
 x'&=&xy'/(x+y+xyy'),\nn\\
 y''&=&x+y+xyy',\nn\\
 y'&=&yy'/(x+y+xyy').\nn
 \eea
Solving the third equation one obtains $yy'=-1.$ Then substituting this to the other equations gives us: $x'=-xy^{-2}, y''=y.$ Other cases are equivalent to one of these due to the symmetry of the electric solution with respect to the change of variables $1\leftrightarrow 3.$ This calculation shows that the space of colorings of a diagram \ref{6-move} is isomorphic to the space of colorings corresponding to the left hand side of the 6-th move.
\begin{figure}[h!]
\centering
\includegraphics[width=8cm]{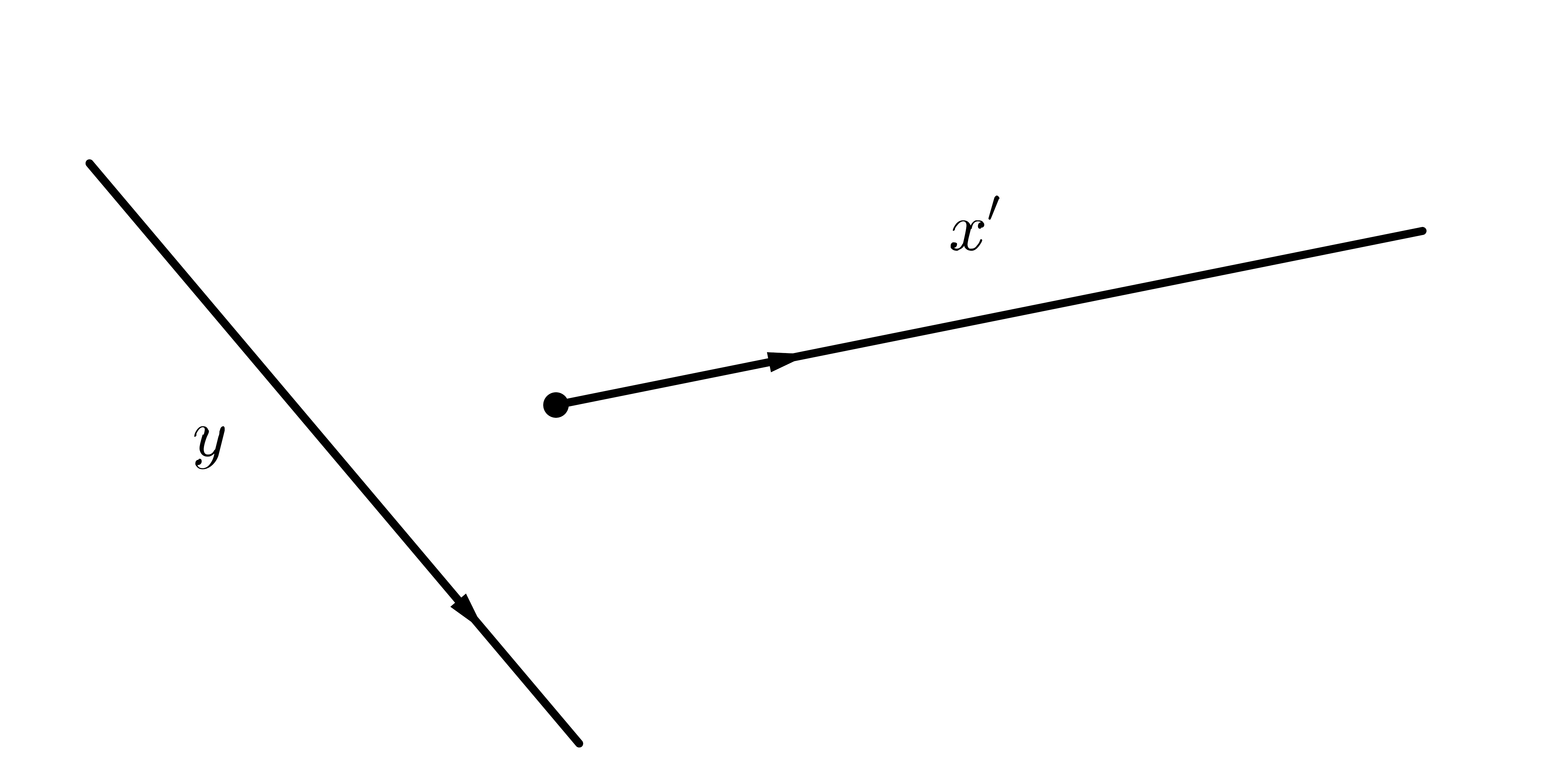}
\caption{Left hand side of the 6-th move}
\label{6-move22}
\end{figure}
Moreover we see that the cocycle $-c_1 c_2$ of lemma 4 of \cite{KST} composed with a character from $\Z/m^\Z$ to $\CC$ or $\R$ is nothing but $-y y'=1$ in both cases. This shows that the partition function remains invariant.

The rest of the proof is concerned with the invariance with respect to the 3-rd and the 7-th moves. It is based on several statements, which we collect in following subsections. 

\subsection{The 3-rd move}

As far as the third move (see figure 6) is concerned, the invariance follows from a simple inspection of definitions. Let us call the colors of the incoming arrows at the diagram \textit{initial colors} and the colors of the outgoing edges \textit{transformed colors}. Now, we can show, that the number of regular colorings does not change, when the graph is changed by the third move, simply by showing that for every set of initial colors there is a unique coloring on both sides of the transformation so that the sets of transformed colors on both sides coincide.

Now it is clear that the signs of two triple points at the right hand side of the picture are opposite, so up to the substitution of $\Phi^{-1}$ for $\Phi$ there are two possibilities: either all three edges of the diagram are oriented in the same way (for instance, from top to bottom on the corresponding picture), or two of them have the same direction, and the third one is opposite. This property does not change, when we pass to the deformed picture (where there are two triple points). In the first case we obtain the following coloring of the deformed graph:
\[
(x,y,z)\stackrel{\Phi}{\to}(x',y',z')\stackrel{\Phi^{-1}}{\to}(x,y,z),
\]
and in the second case we obtain the same coloring, but for transposed map; e.g.:
\[
(x,y,z)\stackrel{\Phi^{t_2}}{\to}(x',y',z')\stackrel{\left({\Phi}^{-1}\right)^{t_2}}{\to}(x,y,z).
\]
(The index of transposed direction depends on the order of the branches.) Thus, there is one regular coloring of the deformed diagram, associated with every coloring $(x,y,z)$ of the undeformed diagram.%}

For the same reason the partition function $\chi_\phi$ of the knot does not change, when the diagram is transformed in accordance with the third move: the value of the cocycle $\phi$ at the new emerging triple points are equal, and come into the formula with opposite signs, so they eliminate each other.

\subsection{The 7-th move}
In this case the transformation of graph is more tricky, so we shall need to work more. As before, we begin with proving that when the tetrahedron map $\Phi$ verifies the conditions, mentioned here, the number of regular colorings does not change, when the graph is changed by the seventh Roseman move.

In our reasonings we shall extensively use the next lemma, describing the basic properties of partial inversion operation:
\begin{Lem}
Let $X$ and $Y$ be two finite sets and $f:X\times Y\to X\times Y,\ g:X\to X$ are two bijections. As above, we shall call the map $f^{t_y}:X\times Y\to X\times Y$ for which   $f^{t_y}(x,y)=(x',y')$ iff $f(x,y')=(x',y)$, (partial) $Y$-inversion of $f$. Then there exists $Y$-inversion of $f$ iff there exists $Y$-inversion of $f(g\times 1)$; in this case we have:
\[
(f(g\times\mathrm{Id}_Y))^{t_y}=f^{t_y}(g\times\mathrm{Id}_Y).
\]
\end{Lem}
\begin{proof}
One can regard $f$ as a map $F:X\to\mathrm{Inj}(Y,\,X\times Y)$ where on the right we have the set of injective maps from $Y$ int the cartesian product $X\times Y$: put $F(x)(y)=f(x,y)$. Composing maps in $\mathrm{Inj}(Y,\,X\times Y)$ with projections, we can write $F(x)=(F(x)_1,F(x)_2)$, where $F(x)_1:Y\to X,\ F(x)_2:Y\to Y$. In this notation 
\[
f(x,y)=(F(x)_1(y),F(x)_2(y)).
\]
An evident necessary condition for the existence of $f^{t_y}$ is the bijectivity of $F(x)_2$ for all $x\in X$: if this is not true, there will either exist $x$, for which one can find $y\in Y$ such that $f(x,y')\ne(x',y)$ for all $y'$, or an $x$, such that there exists $y\in Y$ for which this equation has more than one solution. On the other hand, in case $F(x)_2$ is invertible for all $x$, we put 
\[
f^{t_y}(x,y)= (F(x)_1(F^{-1}(x)_2(y)),F^{-1}(x)_2(y)).
\] 
It is easy to check that for this map the equality which defines the partial inversion holds. Now the equation $(f(g\times\mathrm{Id}_Y))=f^(g\times\mathrm{Id}_Y)$ is checked by a straightforward computation. The opposite claim follows from the above if we replace $g$ with $g^{-1}$ and $f$ with $f(g\times\mathrm{Id}_Y)$.
\end{proof}
As a straightforward corollary of this lemma, by applying it several times to various factors in the 6-fold cartesian product, on which the tetrahedron equation is defined, we obtain equalities, similar to the following one:
\[
\begin{aligned}
(\Phi^{-1})_{356}\Phi^{t_3}_{123}\Phi^{t_3}_{145}\Phi^{t_3}_{246}&=(\Phi_{123}\Phi_{145}\Phi_{246}\Phi_{356})^{t_3t_5t_6}\\
&=(\Phi_{356}\Phi_{246}\Phi_{145}\Phi_{123})^{t_3t_5t_6}=\Phi^{t_3}_{246}\Phi^{t_3}_{145}\Phi^{t_3}_{123}(\Phi^{-1})_{356}.
\end{aligned}
\]
Observe, that the order of composition changes, since we replace $\Phi_{356}$ by its inverse.

Now the rest of the reasoning is based on the following observations: first of all, one can choose the orientations of the planes, forming the tetrahedron so that the corresponding graph of double points will have the standard orientation (see the left hand side of figure \ref{7-th}): to this end we should choose the normals of the right and left triangles at this picture point outwards, while the remaining two normals will point inward.
\begin{figure}[h!]
\centering
\includegraphics[width=15cm]{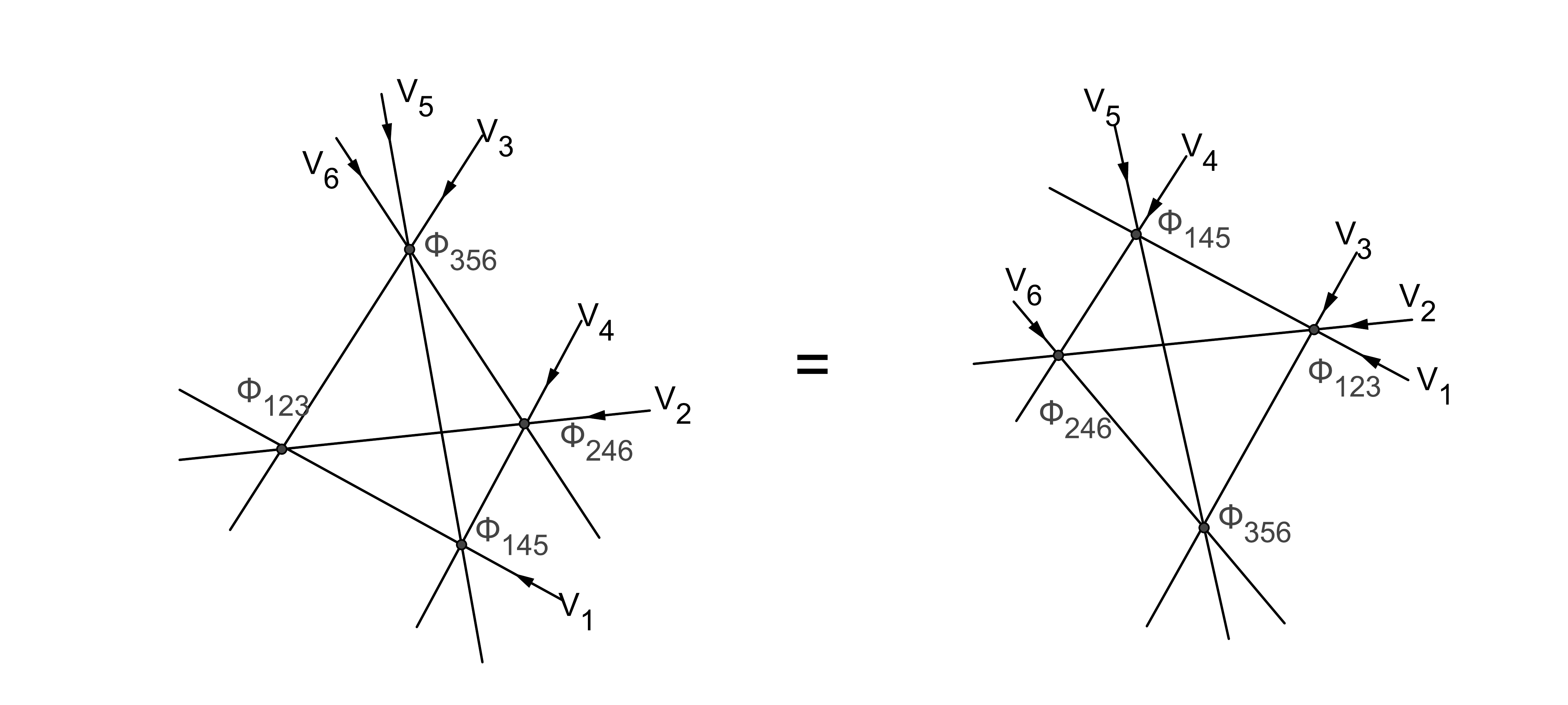}
\caption{7-th move, standard configuration}
\label{7-th}
\end{figure}
As before we call the colors of the incoming arrows \textit{initial colors} and the colors of the outgoing edges \textit{transformed colors}. It is clear, that for every set of 6 initial colors there exists a unique way to choose the colors of all the other edges of the graph at both sides of the figure \ref{7-th}. Thus in this case the bijectivity between the set of colorings before and after the move is provided by the tetrahedron equation. Let us now fix the configuration data in this case: all vertices carry the sign $+$, the order of edges is canonical, i.e. $V_1>V_2>\ldots>V_6$, this corresponds to the lexicographic order of 2-dimensional branches if we mark a branch by the triple of edges transversal to it: $123, 145, 246, 356.$

All other configurations differ from this one by a series of mutations of two types: either switching the orientation of a branch, or by changing the order of two neighbor branches. Suppose for example, that we switch the orientation of the first branch (i.e. the face $356$, containing the edges $1,\,2$ and $4$ at the picture). Choosing the orientations as described in the beginning of previous section, we see that the vertices $123,145,246$ will change their signs and the edges $1,2$ and $4$ will change their direction. The resulting configuration is shown at figure \ref{alt-conf1}.
\begin{figure}[h!]
\centering
\includegraphics[width=15cm]{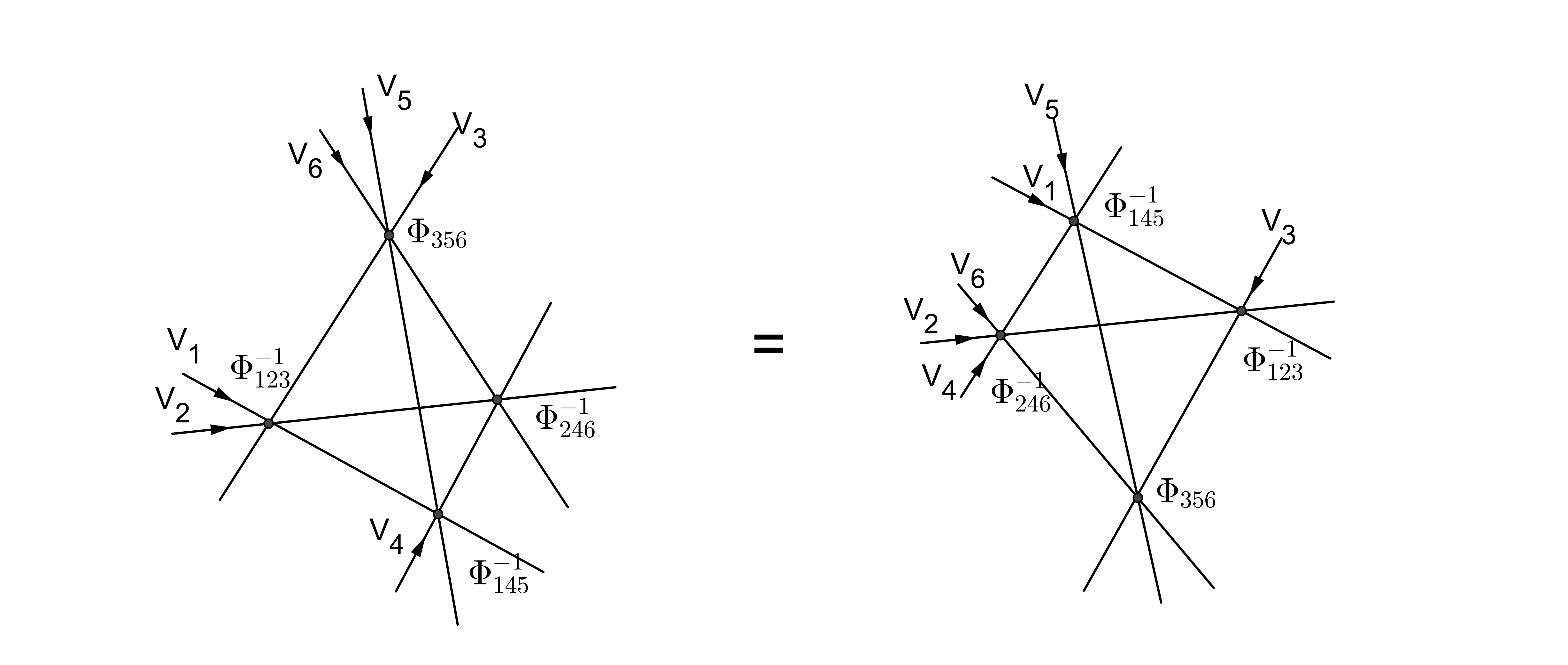}
\caption{7-th move, transposition 356}
\label{alt-conf1}
\end{figure}
Now the statement we need shall follow, if we prove that for any set of initial colors there exists a unique way to color all the edges on both sides of this figure so that the transformed colors are the same. Using the notion of partial inversions we can formalize this condition as a relation on maps from the 6-tuple of initial colors to the 6-tuple of the transformed colors (here we use the notation of the canonical configuration \ref{7-th} and preserve the names of cartesian factors from the original 6-tuple on the triples substituted into $\Phi$):
\[
(\Phi^{t_1,t_2}_{123})^{-1}(\Phi^{t_1,t_4}_{145})^{-1}(\Phi^{t_2,t_4}_{246})^{-1}\Phi_{356}=\Phi_{356}(\Phi^{t_2,t_4}_{246})^{-1}(\Phi^{t_1,t_4}_{145})^{-1}(\Phi^{t_1,t_2}_{123})^{-1}.
\]
This could be simplified:
\bea
\Phi^{t_3}_{123}\Phi^{t_5}_{145}\Phi^{t_6}_{246}\Phi_{356}=\Phi_{356}\Phi^{t_6}_{246}\Phi^{t_5}_{145}
\Phi^{t_3}_{123}.\label{eqiv1}
\eea
Let us demonstrate  that \ref{eqiv1} is a consequence of the tetrahedron equation. To do it consider an equivalent form of the latter:
\bea
(\Phi_{356})^{-1}\Phi_{123}\Phi_{145}\Phi_{246}=\Phi_{246}\Phi_{145}\Phi_{123}(\Phi_{356})^{-1}\nn
\eea
and perform a partial inversion of both sides with respect to the components $3,5$ and $6.$ We start with the left hand side:
\bea
\left((\Phi_{356})^{-1}\Phi_{123}\Phi_{145}\Phi_{246}\right)^{t_3, t_5, t_6}=\left(\Phi_{123}\Phi_{145}\Phi_{246}\right)^{t_3, t_5, t_6}\Phi_{356}=\Phi^{t_3}_{123}\Phi^{t_5}_{145}\Phi^{t_6}_{246}\Phi_{356}.\nn
\eea
On the right hand side we obtain:
\bea
\left(\Phi_{246}\Phi_{145}\Phi_{123}(\Phi_{356})^{-1}\right)^{t_3,t_5,t_6}
=\Phi_{356}\left(\Phi_{246}\Phi_{145}\Phi_{123}\right)^{t_3,t_5,t_6}
=\Phi_{356}\Phi^{t_6}_{246}\Phi^{t_4}_{145}\Phi^{t_3}_{123}.\nn
\eea
The case, when the orientation of the face $135$ is inverted, can be treated in an analogous way: we start with the equivalent form of the TE 
\bea
\Phi_{145}\Phi_{123}\Phi^{-1}_{356}\Phi^{-1}_{246}=\Phi^{-1}_{246}\Phi^{-1}_{356}\Phi_{123}\Phi_{145}\nn
\eea
and perform the partial reversion on components $2,4$ and $6$
\bea
\Phi_{246}\Phi^{t_4}_{145}\Phi^{t_2}_{123}\Phi^{t_3,t_5}_{356}=
\Phi^{t_3,t_5}_{356}\Phi^{t_2}_{123}\Phi^{t_4}_{145}\Phi_{246}.\nn
\eea
Then by conjugating we obtain:
\bea
\Phi^{t_6}_{356}\Phi_{246}\Phi^{t_4}_{145}\Phi^{t_2}_{123}=
\Phi^{t_2}_{123}\Phi^{t_4}_{145}\Phi_{246}\Phi^{t_6}_{356}.\nn
\eea
This is exactly what we need to prove. The other cases are treated in the same way: just observe that the same techniques can be applied recurrently to arbitrary equalities of the sort
\[
A_{ijk}B_{lmn}\dots=\dots B_{lmn}A_{ijk},
\]
where $A,\,B,\dots$ are the maps $X^{\times 3}\to X^{\times 3}$, and $i,j,k,\dots$ are the triples of coordinates, to which they are applied. In particular, in our reasoning they can be different partly-transposed solutions. The important thing is that starting from a correct relation we obtain another correct relation.

Now we consider the mutation consisting in interchanging the neighboring faces. Let us consider the pair of faces: $145$ and $246.$ If we change their order this entails inversion of signs at vertices $123$ and $356$ 
as like as the direction of the edge $3.$ We also have to change the orders of incoming edges at vertices $123$ and $356.$ We illustrate this situation on the picture \ref{7refl}.
\begin{figure}[h!]
\centering
\includegraphics[width=15cm]{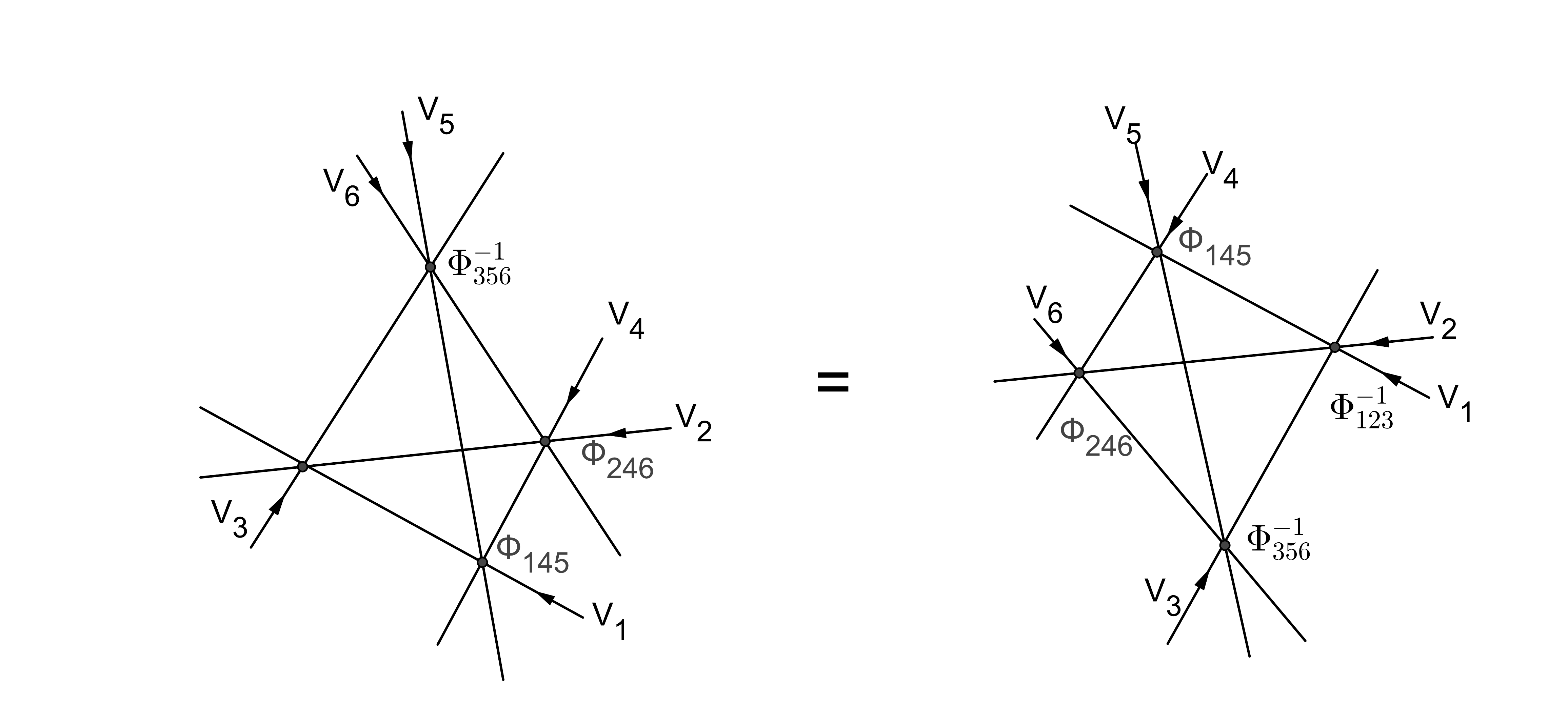}
\caption{7-th move, order change}
\label{7refl}
\end{figure}
Hence we need the following relation in this case:
\bea
\label{eq2}
\left(\Phi_{213}^{t_3}\right)^{-1}\Phi_{145}\Phi_{246}\left(\Phi_{365}^{t_3}\right)^{-1}
=\left(\Phi_{365}^{t_3}\right)^{-1}\Phi_{246}\Phi_{145}\left(\Phi_{213}^{t_3}\right)^{-1}.
\eea
To obtain this let us start with an equivalent form of the TE:
\bea
\Phi_{246}\Phi_{145}\Phi_{123}(\Phi_{356})^{-1}=(\Phi_{356})^{-1}\Phi_{123}\Phi_{145}\Phi_{246}\nn
\eea
then perfrorming the indices change $1\leftrightarrow 2,$ $5\leftrightarrow 6,$
\bea
\Phi_{145}\Phi_{246}\Phi_{213}(\Phi_{365})^{-1}=(\Phi_{365})^{-1}\Phi_{213}\Phi_{246}\Phi_{145}\nn
\eea
and ultimately the partial inversion with respect to the space $3$:
\bea
\Phi_{145}\Phi_{246}\Phi_{365}^{t_5 t_6}\Phi_{213}^{t_3}=\Phi_{213}^{t_3}\Phi_{365}^{t_5 t_6}\Phi_{246}\Phi_{145}.\nn
\eea
This differs from \ref{eq2} by multiplying by $\Phi_{213}^{t_1 t_2}$ on both sides. Other face choices are verified similarly. It is a bit tricky to prove in this way that similar equalities hold in the general case. However, one can pass to a slightly different point of view which will make the one can reason as follows: for every given order of 2-dimensional sheets on the left hand side of the 7-th move one can choose orientations of the sheets such that the corresponding orientations and order of the edges will be canonical. We recall, that the order of edges in a triple point is determined by the order of transversal 2-dimensional sheets, and thus the order of edges of tetrahedron is determined by this rule (up to the exchange of the 3-rd and the 4-th edges which can be ruled out by the condition that edges 1, 2 and 3 meet in one point). Now if the normals of faces 135 and 456 point outside and the rest two normals point inside the tetrahedron, then all the directions of the edges coincide with the standard ones. Thus, up to a permutation of vertices of tetrahedron we can obtain standard configuration of edges for arbitrary order of 2-branches (of course, this means that edges will change their numbers too).  Thus every permutation of sheets can be identified with the change of orientations of 2-dimensional branches in a standard configuration.

Let us proceed to the demonstration of the invariance of the partition function. To do this we need to calculate its part corresponding to the tetrahedral configuration of the whole graph.
 Let us recall that due to lemma \ref{lem_lin} 
\bea
A(e_i\otimes e_j \otimes e_k)=\phi(i,j,k) e_{i'}\otimes e_{j'} \otimes e_{k'}\nn
\eea
is a solution for the matrix TE.
Then the defining relation \ref{cocycle} for the cocycle $\phi$ can be interpreted as an equality of matrix elements for product of matrices $A$ which is equivalent to the matrix tetrahedron equation for the matrix $A.$
\begin{Lem}
\label{lem3}
If the partial transposed maps of $\Phi$ are defined than 
\bea
(A^{t_i})^{-1}=(A^{-1})^{t_i}.\nn
\eea
Here $A^{t_i}$ means the transposition with respect to the $i$-th space.
\end{Lem}
\begin{proof}
Let us verify the statement for $i=3.$ In this case both sides are represented by a linear operator $B$ defined by 
\bea
B(e_{i'}\otimes e_{j'} \otimes e_{k'})=\phi^{-1}(i,j,k') e_i\otimes e_j \otimes e_k\nn
\eea
if $\Phi(i,j,k')=(i',j',k).$
\end{proof}
Hence the invariance of the partition function $Z(s)$ with respect to the standard $7$-th move is equivalent to the matrix tetrahedron equation on $A$ constructed by $\Phi$ and $\phi.$ Our goal is to tame all versions of the $7$-th move and we will proceed as in the case of colorings. First, consider the change of face orientation. In the $356$-face case we need the following:
\bea
(A_{123}^{t_1 t_2})^{-1}(A_{145}^{t_1 t_4})^{-1}(A_{246}^{t_2 t_4})^{-1}A_{356}=A_{356}(A_{246}^{t_2 t_4})^{-1}(A_{145}^{t_1 t_4})^{-1}(A_{123}^{t_1 t_2})^{-1}.\nn
\eea
This is equivalent to
\bea
A_{356}A_{123}^{t_1 t_2}A_{145}^{t_1 t_4}A_{246}^{t_2 t_4}=A_{246}^{t_2 t_4}A_{145}^{t_1 t_4}A_{123}^{t_1 t_2}A_{356}.\nn
\eea
Then applying transposition with respect to the spaces $1,2$ and $4$ we obtain the standard matrix TE: (we write only the left side)
\bea
A_{356}A_{123}^{t_1 t_2}A_{145}^{t_1 t_4}A_{246}^{t_2 t_4} \stackrel{t_4}{\rightarrow}
A_{356}A_{123}^{t_1 t_2}A_{246}^{t_2} A_{145}^{t_1}\stackrel{t_2}{\rightarrow}
A_{356}A_{246}A_{123}^{t_1} A_{145}^{t_1}\stackrel{t_1}{\rightarrow}
A_{356}A_{246}A_{145}A_{123}.\nn
\eea
Let us consider another nontrivial case: the orientation change for the $246$-face. We need:
\bea
(A_{123}^{t_1 t_3})^{-1}(A_{145}^{t_1 t_5})^{-1}A_{246}(A_{356}^{t_3 t_5})^{-1}=(A_{356}^{t_3 t_5})^{-1}A_{246}(A_{145}^{t_1 t_5})^{-1}(A_{123}^{t_1 t_3})^{-1}.\nn
\eea
By conjugating we obtain:
\bea
(A_{246})^{-1}A_{145}^{t_1 t_5} A_{123}^{t_1 t_3} (A_{356}^{t_3 t_5})^{-1}=(A_{356}^{t_3 t_5})^{-1} A_{123}^{t_1 t_3}A_{145}^{t_1 t_5}(A_{246})^{-1}.\nn
\eea
By lemma \ref{lem3} we change the order of operations:
\bea
(A_{246})^{-1}A_{145}^{t_1 t_5} A_{123}^{t_1 t_3} (A_{356}^{-1})^{t_3 t_5}=(A_{356}^{-1})^{t_3 t_5} A_{123}^{t_1 t_3}A_{145}^{t_1 t_5}(A_{246})^{-1}.\nn
\eea
and then perform the consequitive transposition of matrices with respect to the spaces $1,5$ and $3:$
\bea
(A_{246})^{-1}  (A_{356})^{-1}A_{123}A_{145}=A_{145}A_{123}(A_{356})^{-1}(A_{246})^{-1}.\nn
\eea
The last equation is equivalent to the TE.
 
The demonstration of the invariance for the configuration obtained by the face order change exactly reproduces the argument on the space of colorings.

\section{Conclusion}
This section contains few observations and remarks, concerning the possible generalizations and further developments of the theory, presented here. 
\subsection{Framed $2$-knots}
We begin with the following observation: the partition function we describe here is not \textit{a priori} an invariant of the knot itself, but only of the class of its diagrams, namely: it is preserved by all the Roseman moves, except the 2-nd and 4-th ones.

It is not difficult to understand the reason of this phenomenon. Namely as one readily sees, these two moves are principally different from the rest of the list at figure 6. Indeed, these moves do change the topology of the double points graph of the diagram in the following sense: at every triple point of the diagram one can say which pairs of double point edges are complementary (we called such pairs \textit{lines} above). Let us then ``resolve'' every triple point of the graph by saying that it is, in fact equal to three points, one at every line, meeting there. In this way we shall obtain a collection of 1-dimensional manifolds (circles and segments), uniquely defined by the diagram. We shall call this manifold a \textit{resolution of the singular point set}, and denote $S(D)$ (of course, $S(D)$ can be not connected). The following statement is obvious:
\begin{Prop}
The 3-rd, the 5-th, the 6-th and the 7-th Roseman moves do not change the topology of S(D). 
\end{Prop}
Unlike this, the 2-nd movement always changes the number of components of the boundary of $S(D)$, and as for the 4-th movement, it can both change the topology and preserve it. Thus we can hope, that \textit{the partition function, discussed in this paper is invariant of knots diagrams whose $S(D)$ is fixed}.

This situation is quite similar to the theory of framed knot invariants: recall, that the framed knot invariant is an invariant of usual (i.e. 1-dimensional) knots diagrams which is preserved by the 2-nd and the 3-rd Reidemeister moves and some modification of the 1-st move. It can be interpreted as an invariant of the knots, equipped with a normal vector field in $\mathbb R^3$. It would be intriguing to find similar interpretation of the partial invariant we discuss here.

\subsection{Quantum topological field theory in $d=4$}
It is well known \cite{Witten} that there is a close relation between the theory of Jones-Witten invariants  and Chern-Simons quantum topological field theory, namely it turns out that the Jones polynomial coincides with the mean value of the Wilson loop in this theory with the gauge group $SU(2).$ There are many generalizations of this theory with wide spectrum of applications in modern physics. On of them is the BF-theory \cite{BF} whose classical version is given by a Lagrangian:
\bea
S(\omega,B)=tr\int_M B\wedge F \nn
\eea
where $F=d\omega+\omega\wedge \omega$ is the connection curvature, and $B$ is a 2-form on a 4-manifold $M$. 
An imminent problem here is a quest for an analog of the Witten result, i.e. an interpretation of the mean value of the monodromy of the 2-connection $B$ aver a 2-surface in $M$ as an invariant. Up to our knowledge  the mathematical description of such a theory could be considered more or less complete only in the abelian case, when the role of a bundle is played by gerbes, for which there is an adequate algebraic theory. However there are some approaches to nonabelian gerbes \cite{ng}, and some differential geometric description for 2-connections \cite{BS}.

Moreover we suppose that there is a deep relation of this problem with the combinatorial approach to the topological field theories \cite{comb}. In this domain one constructs combinatorial invariants of manifolds, that is some quantities constructed as invariants of simplicial complexes with respect to the Pachner moves. We hope that in analogy with the constructions of the invariants of instrumented 3-manifolds, i.e. manifolds with a tangle, one would be able to construct invariants of instrumented 4-manifolds with embedded 2-surfaces.

\subsection{Regular $3$-d lattices and statistical models}
Consider a periodic three-dimensional lattice of size $K\times L\times M,\ K,L,M\in\mathbb N$; we mark the edges incoming to the node $(i,j,k)$ as $x_{i,j,k},y_{i,j,k},z_{i,j,k}$. The periodicity conditions imply $*_{N+1,j,k}=*_{1,j,k},$ and similar identities for other indexes. Consider a statistical model with the Boltzmann weights in the nodes of the lattice sites determined by the value of the 3-cocycle $\phi$ of the tetrahedral complex. The states are defined as admissible coloring of the edges, i.e. such that in each node the condition fulfills:
\bea
\Phi(x_{i,j,k},y_{i,j,k},z_{i,j,k})=(x_{i+1,j,k},y_{i,j+1,k},z_{i,j,k+1}).
\eea
A partition function is defined as follows:
\bea
Z(s)=\sum_{Col}\prod_{i,j,k}\phi(x_{i,j,k},y_{i,j,k},z_{i,j,k})^s.
\eea
\\
To explore the ''integrability" of the subsidiary quantum problem one needs the layer-to-layer transfer-matrix. In order to determine what it is, we need another interpretation of the partition function.
\\
We associate a copy of the space $V$ to each line of the lattice. For convenience, we denote the vertical spaces by characters $V_{ik}$ and the horizontal ones - by $E_i$ and $N_k.$
\\
We construct an operator $A_{ik}(s)$ with the chosen $3$-cocycle according to lemma \ref{lem_lin} statement.
\begin{figure}[h]
\centering
\includegraphics[width=120mm]{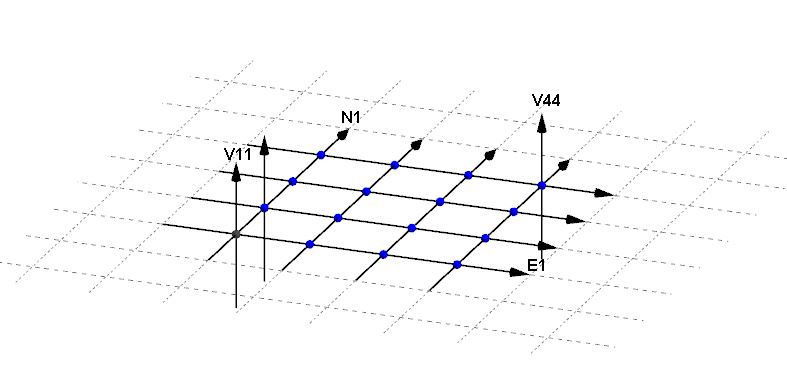} 
\caption{1-layer configuration}
\end{figure}
\\
Let us define the transfer-matrix by a $1$-layer product:
\bea
T(s)=Tr\prod_i\prod_k A_{ik}(s).\nn
\eea
The product and trace of matrices in the formula are taken with respect to horizontal spaces. This operator acts on the tensor product of vertical spaces. Here $A_{ik}(s)$ is an operator on the space $E_i\otimes V_{ik}\otimes N_k$. It turns out that the partition function takes the form
\bea
Z(s)=Tr_{V_{jk}} T(s)^L.\nn
\eea

Such issues as the asymptotic behavior of partition functions when the size of the lattice grows may be solved by the study of the spectrum of the transfer-matrix. The integrability condition, i.e. the possibility of including the transfer-matrix in a large commutative family, simplifies the problem of finding the spectrum.


\begin{thebibliography}{50}
\bibitem{T}
D.V. Talalaev, {\em Zamolodchikov tetrahedral equation and higher Hamiltonians of 2d quantum integrable systems.}  arXiv:1505.06579

\bibitem{CES} S. Carter, M. Elhamdadi, M. Saito {\em Homology Theory for the Set-Theoretic Yang-Baxter Equation and Knot Invariants from Generalizations of Quandles}, arXiv:0206255

\bibitem{Zam} 
A. B. Zamolodchikov, {\em Tetrahedra equations and integrable systems in three-dimensional
space}, Zh. Eksp. Teor. Fiz. {\bf 79} (1980) 641Â664. [English translation: Soviet Phys. JETP {\bf 52} (1980) 325-326].

\bibitem{KST}
I. Korepanov, G. Sharygin, D. Talalaev {\em Cohomologies of $n$-simplex relations.}  arXiv:1409.3127 

\bibitem{W}
E. Witten, {\em Quantum Field Theory and the Jones Polynomial.} Commun. Math. Phys. {\bf 121},351-399 (1989)

\bibitem{CSJ}
Carter, J.S., Jelsovsky, D., Langford, L., Kamada, S., Saito, M., {\em Quandle cohomology and
state-sum invariants of knotted curves and surfaces,} Trans. Amer. Math. Soc. {\bf 355} (2003), no.
10, 3947-3989.

\bibitem{Zeeman}
E. C. Zeeman, {\em Twisting spun knots.} Trans. Amer. Math. Soc. {\bf 115} (1965), 471{495. MR
33:3290}

\bibitem{Roseman98}
D. Roseman, {\em Reidemeister-type moves for surfaces in four-dimensional space.} Knot theory, Banach center publications, Vol {\bf 42}, Institute of mathematics, Polish academy of sciences. Warszawa 1998.

\bibitem{Mat82}
S. V. Matveev, {\em Distributive groupoids in knot theory}, Mathematics of the USSR-Sbornik(1984),{\bf 47}(1):73 

\bibitem{Witten}
E. Witten, {\em Quantum Field Theory and the Jones Polynomial, } Commun. Math. Phys. {\bf 121}, 351-399 (1989)

\bibitem{BF}
A. Cattaneo, P. Cotta-Ramusino, J. Frohlich, M. Martellini, {\em Topological BF theories in 3 and 4
dimensions,} hep-th/9505027

P. Mnev	{\em Notes on simplicial BF theory,} Mosc. Math. J., 2009, {\bf 9},	N2, pp 371-410

\bibitem{ng}
P. Aschieri, L. Cantini, and B. Jurco, {\em Nonabelian bundle gerbes, their differential geometry and
gauge theory} . Commun. Math. Phys., {\bf254}: 367-400, 2005. [arxiv:hep-th/0312154]

\bibitem{BS}
J. Baez, U. Schreiber, {\em Higher Gauge Theory: 2-Connections on 2-Bundles,} arXiv:hep-th/0412325

\bibitem{comb}
I. Korepanov, {\em Geometric torsion and topological field theories a la Atyah,} Theoretical and mathematical physics, 2009, {\bf 158}, 3, 405-418


\end{thebibliography}
\end{document}